\documentclass[11pt,onecolumn]{ieeetran}

\usepackage{amsmath,amssymb,amsthm}
\usepackage{mathtools}
\usepackage{enumitem}
\usepackage{hyperref}

\hypersetup{colorlinks=true,allcolors=blue}

\newcommand{\N}{\mathbb{N}}
\newcommand{\Q}{\mathbb{Q}}
\newcommand{\R}{\mathbb{R}}
\newcommand{\E}{\mathbb{E}}
\newcommand{\Prob}{\mathbb{P}}

\newtheorem{theorem}{Theorem}
\newtheorem{lemma}{Lemma}

\newtheorem{corollary}{Corollary}
\theoremstyle{definition}
\newtheorem{definition}{Definition}
\newtheorem{remark}{Remark}
\newcommand{\orcid}[1]{\href{https://orcid.org/#1}{ {\includegraphics[scale=0.5]{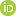}}}}
\title{Computability Limits of Sequential Hypothesis Testing}

\author{Amir Leshem\thanks{Amir Leshem is with Faculty of Engineering, Bar-Ilan University, 52900, Ramat-Gan, Israel. The research was partially supported by ISF grant 2197/22. amir.leshem@biu.ac.il}\textsuperscript{ 
  \orcid{0000-0002-2265-7463}}, ~\IEEEmembership{Fellow,~IEEE}
  }
\date{\today}

\begin{document}
\maketitle
\begin{abstract}
Sequential hypothesis testing asks for decision rules that update as data arrive. A natural goal is
\emph{eventual correctness}: the rule may change its mind early on, but it should make only finitely
many wrong decisions almost surely. Starting from Cover's theorem, which guarantees such behavior
for membership in a countable set of candidate means, we ask a sharper question: \emph{which sets
actually admit computable sequential decision procedures with finitely many errors?} We answer this
optimally by giving a complete characterization both necessary and sufficient of the subsets of
$\Q$ that admit a computable finite-error sequential membership test. We further extend the
characterization to any \emph{effectively presented} countable family of real means, exactly the
setting in which Cover's identification rule can be implemented computably. Beyond the technical
boundary, the results clarify within a precise probabilistic setting what it can mean for inquiry
to ``converge to the truth,'' and they formalize a limit to which empirical methods can be expected
to succeed when only eventual stabilization (rather than fixed-time guarantees) is demanded.

keywords: Cover's theorem, sequential decision procedures, finite error learning, limit computability,
$\Delta^0_2$ sets.
\end{abstract}
\section{Introduction}
In many scientific settings, hypotheses are not settled by a single finite experiment.
Rather, one adopts a working model and revises it when accumulating evidence forces a change.
This familiar ``Popperian'' picture is naturally idealized by a \emph{limit test}: a procedure that receives an increasing data record and may change its verdict, but is required to stabilize eventually (equivalently, it makes only finitely many mistakes along the realized data stream).
From a computability-theoretic perspective, such tests correspond to \emph{limit computation} (computation with finitely many mind changes).
In recursion theory this appears as ``trial-and-error'' or limiting computation \cite{putnam1965}, while in inductive inference it is studied as identification/learning in the limit 
\cite{gold1967,terwijn2006,jain1999}.
A classical theorem of Cover~\cite{Cover1973} concerns sequential decision procedures for the mean $\mu$ of an i.i.d.\ process with finite second moment.
Given any prescribed countable set $S=\{s_1,s_2,\dots\}\subseteq\R$, Cover constructs an explicit sequential rule which, with probability one, eventually identifies the correct hypothesis $\mu=s_i$ whenever $\mu\in S$, and otherwise eventually declares $\mu\notin S$ for all $\mu\notin S$ outside a Lebesgue-null exceptional set.
The case $S=\Q$ yields an ``irrationality test'' for the mean.

The purpose of this paper is to show that Cover's construction is closed under \emph{limit-computable} post-processing on the identified index.
Fix a computable enumeration $e:\N\to\Q$ and let $A\subseteq\Q$ be presented via its index set $I_A=\{i\in\N: e(i)\in A\}$.
If $I_A$ is limit computable (equivalently, $I_A\in\Delta^0_2$), then membership in $A$ can be decided by a computable limit test built from two ingredients:
(i) Cover's sequential identification of the rational index when $\mu\in\Q$, and
(ii) a limit-computable approximation of $I_A$ (by Shoenfield's limit lemma \cite{Shoenfield1993}).
Conversely, any computable limit test that succeeds for every rational mean induces a limit approximation to $I_A$, and hence forces $I_A$ to be limit computable.
In particular, these results show that, even when driven by random data, sequential procedures that are allowed finitely many revisions can decide exactly the limit-computable ($\Delta^0_2$) properties of a countable, effectively presented hypothesis class up to the unavoidable Lebesgue-null exceptional set in Cover's theorem for means outside the hypothesis class.

Finally, we note a methodological point about effectivity.
Although our procedures are phrased using real-valued statistics (sample means and variances), at each stage they require only finitely many rational comparisons (for instance, checking whether a candidate lies in a data-dependent open interval).
Thus the decision rules can be implemented from sufficiently accurate rational approximations of the observed values, without committing to a fully effective probability framework. Viewed as a learning problem, our results give an exact computability-theoretic boundary for when mean-membership hypotheses are learnable with almost-surely finitely many errors (i.e., eventual stabilization) by a computable sequential procedure.
\subsection*{Contribution and novelty.}
The main contribution of this paper is a sharp computational characterization of sequential hypothesis tests over countable mean hypotheses.
While Cover's theorem is often presented as a rationality test for the mean, our results isolate the exact computational strength of Cover-style sequential identification.
Specializing Cover's countable-set identification procedure to $S=\mathbb{Q}$ and composing it with a limit approximation (via the Shoenfield limit lemma), we show that every limit-computable property of the \emph{identified rational index} is limit decidable from the data.
Conversely, we prove that no computable limit test can decide membership in a set $A\subseteq\mathbb{Q}$ for all rational means unless the associated index set $I_A$ is itself limit computable (equivalently, $\Delta^0_2$).
Together, these results yield an ``if and only if'' boundary: for rational means, $A\subseteq\mathbb{Q}$ admits a computable Popper-style sequential test precisely when $I_A$ is limit computable.

\section{Problem formulation}

Let $X_1,X_2,\dots$ be independent and identically distributed real-valued random variables with
finite second moment and mean $\mu = \E[X_1]$.
A \emph{sequential test} is a sequence of functions
\[
F_n : \R^n \to \{0,1\},
\qquad n\ge 1,
\]
interpreted as provisional decisions based on the first $n$ observations.

\begin{definition}[Finite mistakes / stabilization]\label{def:finite-mistakes}
Fix a truth value $\theta\in\{0,1\}$. We say that a sequential test $(F_n)$
\emph{makes only finitely many mistakes almost surely} for parameter $\mu$
if
\[
\Prob_\mu\Big( F_n(X_1,\dots,X_n)\neq \theta \text{ for infinitely many }n\Big)=0.
\]
Equivalently, $\Prob_\mu\big(\lim_{n\to\infty}F_n(X_1,\dots,X_n)=\theta\big)=1$.
\end{definition}

\paragraph{The decision problem.}
Fix a subset $A\subseteq\Q$. We wish to decide whether $\mu\in A$.
Since $A\subseteq\Q$, the correct answer is automatically \emph{no} for irrational $\mu$.
Thus the main requirement is strong correctness for all rational means,
while allowing a measure-zero exceptional set of irrationals, as in Cover's theorem.

\paragraph{Effectivity.}
Cover's procedures are specified by explicit arithmetic operations on the observed samples
(sample means, sample variances, and comparisons to rational thresholds).
The decision rule uses only finitely many rational comparisons at each stage; hence, it can be implemented from sufficiently accurate rational approximations of the data.
\section{Effective enumerations and \texorpdfstring{$\Delta^0_2$}{Delta-2} subsets of \texorpdfstring{$\Q$}{Q}}

\begin{definition}[Computable enumeration of $\Q$]
A \emph{computable enumeration} of $\Q$ is a total computable \emph{one-to-one} function onto $\Q$, 
$e:\N\to\Q$. Write $q_i:=e(i)$.
\end{definition}

Given $A\subseteq\Q$ and a fixed enumeration $e$, define the \emph{index set}
\[
I_A := \{\, i\in\N : q_i \in A \,\} \subseteq \N.
\]

\begin{definition}[$\Delta^0_2$ subset of $\Q$ relative to an enumeration]\label{def:delta2Q}
Let $e:\N\to\Q$ be fixed.
A set $A\subseteq\Q$ is called \emph{$\Delta^0_2$ relative to $e$} if its index set $I_A\subseteq\N$
is a $\Delta^0_2$ set of integers (equivalently, both $I_A$ and its complement are $\Sigma^0_2$).
\end{definition}

The key computability-theoretic tool is the characterization of $\Delta^0_2$ sets as those decidable
\emph{in the limit}.

\begin{lemma}[Shoenfield Limit Lemma]\label{lem:shoenfield}
A set $B\subseteq\N$ is $\Delta^0_2$ iff there exists a total computable function
$a:\N\times\N\to\{0,1\}$ such that for every $i\in\N$,
\[
\mathbf{1}_B(i) \;=\; \lim_{s\to\infty} a(i,s),
\]
i.e.\ for each fixed $i$, the values $a(i,s)$ change only finitely often and then stabilize to $\mathbf{1}_B(i)$.
\end{lemma}

\begin{remark}
\label{rem:presentation}
An  effective procedure defining $A$ must be given by some finite code.
In Theorem~\ref{thm:main} below, we assume access to an index for a computable approximation
$a(i,s)$ of the kind provided by Lemma~\ref{lem:shoenfield}.
Equivalently we may assume an effective $\Sigma^0_2/\Pi^0_2$ presentation and invoke Shoenfield's Lemma to obtain a uniformly computable approximation $a$)
\end{remark}

\section{Cover's countable-mean test specialized to $\Q$}

Cover~\cite{Cover1973} studies the hypothesis testing problem for a countable set $S=\{p_1,p_2,\dots\}\subseteq\R$:
\[
H_i:\ \mu = p_i\ (i\ge 1) \qquad \text{vs.} \qquad H_0:\ \mu\notin S,
\]
under the success criterion of Definition~\ref{def:finite-mistakes}.
He gives an explicit online decision rule based on open intervals around the sample mean and an
Occam-style complexity threshold.
We recall the main ingredients, specialized to $S=\Q$ and to a fixed computable enumeration $e(i)=q_i$.

\subsection{Decision variables}

Let
\[
\bar X_n := \frac{1}{n}\sum_{t=1}^n X_t,
\qquad
s_n^2 := \frac{1}{n}\sum_{t=1}^n (X_t-\bar X_n)^2
\]
denote the sample mean and (uncorrected) sample variance.

Fix a parameter $\alpha>0$ and define an interval radius
\begin{equation}\label{eq:delta}
\delta_n \;:=\;\max\left\{(1+\alpha)\sqrt{\frac{2 s_n^2 \log\log n}{n}},2^{-n}\right\}
\qquad (n\ge 3).
\end{equation}
The off-target $0$-output guarantee is stated under the usual nondegeneracy assumption
$\hbox{Var}(X_1)>0$ (equivalently $\mathbb{P}(X_1=\mu)<1$).
The auxiliary floor $2^{-n}$ in~\eqref{eq:delta} is added only to make the construction total when $s_n=0$ (the degenerate case), which is also used in the reduction for the necessity direction. Under $\hbox{Var}(X_1)>0$, we have $s_n\to\sigma>0$ almost surely, hence for all sufficiently large $n$ the LIL term dominates $2^{-n}$ and the floor does not affect the asymptotic behavior of the rule.
Cover \cite{Cover1973} shows (using the law of the iterated logarithm together with the almost sure consistency of $s_n^2$) that for each fixed $\mu$,
\begin{equation}\label{eq:lil-inclusion}
|\bar X_n-\mu|<\delta_n \quad \text{for all but finitely many }n,\ \text{almost surely.}
\end{equation}
Then he uses~\eqref{eq:lil-inclusion} to guarantee that the shrinking open interval
$(\bar X_n-\delta_n,\bar X_n+\delta_n)$ eventually contains the true mean $\mu$.

Next define the \emph{least-index rational in the interval}:
\begin{equation}\label{eq:i}
i(t,\delta) \;:=\; \min\{\, i\ge 1 : |q_i-t|<\delta \,\}.
\end{equation}
Since $\Q$ is dense and $\delta>0$, the set in~\eqref{eq:i} is nonempty, so $i(t,\delta)$ is always finite.
Because $e$ is computable and $(t,\delta)$ will be computable from the data,
$i(t,\delta)$ can be found effectively by searching $i=1,2,3,\dots$ until the first hit.

Finally, choose a \emph{decision schedule} and an increasing \emph{complexity threshold}.
A convenient explicit choice is:
\begin{equation}\label{eq:schedule}
n(j) := \lceil j^{p}\rceil,\qquad k_j := j,
\qquad j=1,2,3,\dots,
\end{equation}
for some fixed $p>4$ (Cover suggests $p>6$).

\subsection{Cover's decision rule (for \texorpdfstring{$S=\Q$}{S=Q})}

At decision time $n(j)$:
\begin{enumerate}[label=\textbf{Step \arabic*:},leftmargin=2.4em]
\item Compute $\bar X_{n(j)}$ and $s_{n(j)}^2$ from $X_1,\dots,X_{n(j)}$.
\item Form $\delta_{n(j)}$ as in~\eqref{eq:delta}.
\item Compute $i_j := i(\bar X_{n(j)},\delta_{n(j)})$ as in~\eqref{eq:i}.
\item Output the hypothesis index
\[
C_{n(j)} :=
\begin{cases}
i_j, & \text{if } i_j\le k_j,\\
0, & \text{if } i_j>k_j.
\end{cases}
\]
\end{enumerate}
For times $n(j)\le n < n(j+1)$, keep the output constant: $C_n := C_{n(j)}$.

\begin{theorem}[Cover's countable-mean theorem, specialized to $\ \Q$]\label{thm:coverQ}
Let $X_1,X_2,\dots$ be i.i.d.\ with finite second moment and mean $\mu$.
There exists a Lebesgue-null set $N_0\subseteq \R\setminus\Q$ such that the above procedure satisfies:
\begin{enumerate}[label=(\roman*)]
\item For every $\mu\in\Q$, the output $C_n$ eventually stabilizes almost surely on the correct index $i$
with $\mu=q_i$.
\item For every $\mu\in (\R\setminus\Q)\setminus N_0$, the output $C_n$ eventually stabilizes almost surely on $0$.
\end{enumerate}
In particular, $C_n$ makes only finitely many mistakes almost surely for each $\mu\in\Q$ and for each
$\mu\in(\R\setminus\Q)\setminus N_0$.
\end{theorem}

\subsection{Sketch of Cover's argument (why the null set appears)}\label{sec:cover-sketch}

For completeness, we record the two key ideas in Cover's proof.

\paragraph{Correctness when $\mu\in\Q$.}
Assume $\mu=q_{i^\ast}$.
By~\eqref{eq:lil-inclusion}, eventually $\mu$ lies in the open interval, so $i(\bar X_n,\delta_n)$ is well-defined
and (for large $n$) cannot exceed $i^\ast$.
Because $\delta_n\to 0$, eventually no rational $q_i$ with $i<i^\ast$ can lie in the interval, so
$i(\bar X_n,\delta_n)=i^\ast$ for all but finitely many $n$ almost surely.
Since $k_j\to\infty$, the threshold does not block the true index for large $j$.

\paragraph{Rejection when $\mu\notin\Q$.}
Fix $k\in\N$ and $\delta>0$ and consider the set of parameters for which a ``low-complexity'' rational lies nearby:
\[
E(k,\delta) := \{\mu\in\R : i(\mu,\delta)\le k\}.
\]
Because $i(\mu,\delta)\le k$ means that $\mu$ lies within $\delta$ of at least one of $q_1,\dots,q_k$,
a simple measure bound yields
\begin{equation}\label{eq:measure-bound}
\lambda(E(k,\delta)) \le 2k\delta,
\end{equation}
where $\lambda$ is Lebesgue measure.
Cover chooses the decision times $n(j)$ and thresholds $k_j$ so that
\begin{equation}\label{eq:borel-cantelli-sum}
\sum_{j=1}^\infty \lambda(E(k_j, c\,\delta_{n(j)})) < \infty
\end{equation}
for a suitable constant $c>0$; using~\eqref{eq:measure-bound} this follows from a condition of the form
$\sum_j k_j\delta_{n(j)}<\infty$.
By the Borel-Cantelli lemma, for almost every $\mu$ (i.e.\ outside a null set $N_0$) the event
$\mu\in E(k_j, c\,\delta_{n(j)})$ happens only finitely often.
Together with the LIL control $|\bar X_{n(j)}-\mu|\le \delta_{n(j)}$ eventually,
this implies $i(\bar X_{n(j)},\delta_{n(j)})>k_j$ for all large $j$, hence the procedure outputs $0$ eventually.

We refer to~\cite{Cover1973} for the full details and for specific admissible choices of $n(j)$ and $k_j$.
\begin{remark}[Cover's idealized baseline]\label{rem:cover-baseline}
Cover's construction and analysis are carried out under the usual probabilistic idealization that
real arithmetic and comparisons (e.g.\ forming $\bar X_n$ and testing membership in real intervals)
are available exactly. We take this as the baseline setting for stating the identification property.
Section~\ref{sec:readouts} then shows that the same decision scheme admits a fully computable
realization from finite-precision rational readouts, without changing the probabilistic conclusions
(up to the harmless inflation $\delta_n'=\delta_n+\eta_n$).
\end{remark}

\section{Computability via finite-precision readouts}\label{sec:readouts}

Cover's construction \cite{Cover1973} is formulated over real-valued samples and does not address
effectivity. 

To discuss computability without introducing representations of real numbers, we assume the data
arrive as finite-precision measurements: for each $i$ we observe a rational $\tilde X_i$ approximating
$X_i$, and at time $n$ the procedure receives $(\tilde X_1,\dots,\tilde X_n)$. This matches the usual
sequential-sampling paradigm (new samples over time, no revision of the past) and lets us treat the
test as a computable function on $\Q^n$. We now define what is necessary in order to obtain an effective version of Cover's test.

\subsection{Finite-precision observation model (fixed readouts)}
Let $(X_i)_{i\ge 1}$ be i.i.d.\ real-valued random variables with mean $\mu\in\R$.
Fix a computable error schedule $(\varepsilon_i)_{i\ge 1}$ with $\varepsilon_i\downarrow 0$
(e.g.\ $\varepsilon_i=2^{-i}$).

\begin{definition}[Fixed rational readouts]\label{def:fixed-readouts}
A \emph{finite-precision readout sequence} for $(X_i)$ with accuracy $(\varepsilon_i)$ is a rational
sequence $(\tilde X_i)_{i\ge 1}\subseteq\Q$ satisfying
\begin{equation}\label{eq:fixed-readout-bound}
|\tilde X_i-X_i|\ \le\ \varepsilon_i,\qquad \text{for all }i\ge 1.
\end{equation}
At time $n$, the procedure receives the rational vector $(\tilde X_1,\dots,\tilde X_n)\in\Q^n$.
\end{definition}

Write $\bar X_n=\frac1n\sum_{i=1}^n X_i$ and $\overline{\tilde X}_n=\frac1n\sum_{i=1}^n \tilde X_i$.
Then \eqref{eq:fixed-readout-bound} implies the deterministic mean bound
\begin{equation}\label{eq:mean-bound-fixed}
\bigl|\overline{\tilde X}_n-\bar X_n\bigr|
\ \le\ \frac1n\sum_{i=1}^n \varepsilon_i\ :=:\ \eta_n.
\end{equation}
We now record a convenient choice of the readout tolerances ensuring that the inflation
$\delta_n'=\delta_n+\eta_n$ preserves the summability conditions required in Cover's argument.
\begin{remark}[Choosing $(\varepsilon_i)$ and $\eta_n$ for $k_j=j$, $n(j)=\lceil j^p\rceil$]\label{rem:choose-eps-eta}
Assume fixed readouts $|\tilde X_i-X_i|\le \varepsilon_i$ with a computable tolerance schedule
$\varepsilon_i\downarrow 0$, and set
\[
\eta_n:=\frac1n\sum_{i=1}^n \varepsilon_i,
\qquad\text{so that}\qquad
|\tilde X_n-\bar X_n|\le \eta_n.
\]
We implement mean-threshold steps using $\delta_n'=\delta_n+\eta_n$. For the schedule
$k_j=j$ and $n(j)=\lceil j^p\rceil$ with $p>4$, it suffices that
\[
\sum_{j=1}^\infty j\,\eta_{n(j)}<\infty
\quad\text{in addition to}\quad
\sum_{j=1}^\infty j\,\delta_{n(j)}<\infty.
\]
A simple explicit choice is $\varepsilon_i=2^{-i}$ (or $\varepsilon_i=i^{-2}$), for which
$\eta_n\le C/n$. Hence
\[
\sum_{j=1}^\infty j\,\eta_{n(j)}
\ \le\ C\sum_{j=1}^\infty \frac{j}{n(j)}
\ \le\ C\sum_{j=1}^\infty j^{1-p}
\ <\ \infty
\qquad(\text{since }p>4).
\]
Thus the readout inflation does not affect the validity of the Borel-Cantelli argument.
\end{remark}

\begin{remark}[No rereading and finite memory]\label{rem:no-rereading}
In this model each sample is observed once, at finite precision, and never updated:
nothing about the past changes. The sequential test operates on the discrete stream
$\tilde X_1,\tilde X_2,\dots$ and can be implemented with finite memory (e.g.\ by maintaining the
running rational mean $\overline{\tilde X}_n$ and the current candidate index), without storing
the entire history.
\end{remark}

\subsection{Computable sequential tests (discrete notion)}
\begin{definition}[Computable sequential test from readouts]\label{def:computable-test}
A \emph{computable sequential test} is a sequence of (classical) computable functions
\[
F_n:\Q^n\to\{0,1\}\qquad(n\ge 1).
\]
Given fixed rational readouts $(\tilde X_i)_{i\ge 1}$ as in Definition~4, the realized outputs are
$F_n(\tilde X_1,\dots,\tilde X_n)$.
\end{definition}

\subsection{A robustness lemma for interval-based decisions}
The procedures in this paper ultimately decide membership by locating the sample mean in
shrinking open intervals around candidate points. The next lemma shows that such decisions are
stable under \emph{readout} error once the mean is separated from the interval boundary by more
than $\eta_n$.
\begin{lemma}[Readout stability for open intervals]\label{lem:readout-stability}
Let $I=(a,b)$ be an open interval. Assume the fixed readout model of Definition~4, so that
\eqref{eq:mean-bound-fixed} holds. If
\[
\min\{|\bar X_n-a|,\ |\bar X_n-b|\}>\eta_n,
\]
then
\[
\mathbf{1}\{\tilde X_n\in I\}=\mathbf{1}\{\bar X_n\in I\}.
\]
\end{lemma}

\begin{proof}
By \eqref{eq:mean-bound-fixed}, $|\tilde X_n-\bar X_n|\le \eta_n$. If
$\min\{|\bar X_n-a|,|\bar X_n-b|\}>\eta_n$, then $\bar X_n$ lies at least $\eta_n$ away from both
endpoints, so a perturbation of size at most $\eta_n$ cannot move it across the boundary of $I$.
\end{proof}

\subsection{Readout implementation of Cover-type rules}
Whenever an argument uses a threshold $\delta_n$ for the sample mean, we implement it from readouts
by inflating the threshold to absorb readout error:
\begin{equation}\label{eq:delta-prime}
\delta_n'\ :=\ \delta_n+\eta_n.
\end{equation}
Then the implication
\[
|\bar X_n-q|<\delta_n\ \Longrightarrow\ |\overline{\tilde X}_n-q|<\delta_n'
\]
holds deterministically by \eqref{eq:mean-bound-fixed}. Thus every step of the Cover/Cover-type index
selection based on conditions of the form $|\bar X_n-q|<\delta_n$ (or $\bar X_n\in(q-\delta_n,q+\delta_n)$)
admits a computable readout implementation obtained by replacing $\bar X_n$ with $\overline{\tilde X}_n$
and $\delta_n$ with $\delta_n'$.

\begin{remark}[Null sets and summability conditions]\label{rem:summability-delta-prime}
In the Borel-Cantelli part of Cover's analysis one typically requires a summability condition of the
form $\sum_j k_j\,\delta_{n(j)}<\infty$ (up to constants), using bounds such as
$\lambda(E(k,\delta))\le 2k\delta$ for suitable exceptional sets $E(k,\delta)$.
Under the readout model, the same argument goes through with $\delta_n$ replaced by
$\delta'_n=\delta_n+\eta_n$. In particular, it suffices to choose schedules $(n(j),k_j)$ and
$(\delta_n)$ so that $\sum_j k_j\,\delta'_{n(j)}<\infty$; this is automatic, for example, if
$\sum_j k_j\,\delta_{n(j)}<\infty$ and  $\sum_j k_j\,\eta_{n(j)}<\infty$.
\end{remark}

\subsection{Effect on the main theorems}
All subsequent constructions define tests by composing an index selector (Cover-type rule) with a
$\Delta^0_2$ approximation $a(i,s)$ (e.g.\ via Shoenfield). Under
Definition~\ref{def:computable-test}, these tests are plainly computable because they are maps
$\Q^n\to\{0,1\}$ applied to the rational readout inputs. The correctness proofs are unchanged
except for the uniform replacement $\delta_n\mapsto\delta_n'$ as in \eqref{eq:delta-prime}, justified
by Lemma~\ref{lem:readout-stability}.
In particular, the necessity direction (extracting a $\Delta^0_2$ approximation by feeding constant
rational data) becomes entirely discrete: one evaluates $F_n(q,\dots,q)$ on exact rational inputs.

\begin{remark}[Convention: suppressing the readout map]\label{rem:readout-convention}
The probabilistic model is defined in terms of the real-valued i.i.d.\ samples $(X_i)$ with mean
$\mu$, while the decision rules are computed from the rational readouts $(\tilde X_i)$ of
Definition~4 satisfying \eqref{eq:fixed-readout-bound} (and hence \eqref{eq:mean-bound-fixed}).
To avoid notational clutter, we henceforth suppress the readout map in the arguments of computable
functions and write
\[
F_n(X_1,\dots,X_n)\quad\text{as shorthand for}\quad F_n(\tilde X_1,\dots,\tilde X_n),
\]
and similarly for $C_n$.
\end{remark}

\section{Main result: \texorpdfstring{$\Delta^0_2$}{Delta-2} subsets of \texorpdfstring{$\Q$}{Q}}

We now lift Theorem~\ref{thm:coverQ} from the set $\Q$ itself to arbitrary $\Delta^0_2$ subsets of $\Q$. This is interesting, since it implies that given a Shoenfield representation of a set $A\subseteq \Q$ we can computably design a sequence of decisions converging to the right decision even for more complex sets in the Turing Hierarchy.
\begin{theorem}[Main Theorem]\label{thm:main}
Fix a computable one-to-one enumeration $e:\N\to\Q$ and a set $A\subseteq\Q$ whose index set $I_A\subseteq\N$ is $\Delta^0_2$.
Let $a:\N\times\N\to\{0,1\}$ be a total computable approximation such that
$\mathbf{1}_{I_A}(i)=\lim_{s\to\infty} a(i,s)$ for all $i$ (as in Lemma~\ref{lem:shoenfield}).

Then there exists a computable sequential test $(F_n)$ such that:
\begin{enumerate}[label=(\roman*)]
\item For every $\mu\in\Q$,
\[
\Prob_\mu\!\Big(\lim_{n\to\infty} F_n(X_1,\dots,X_n)=\mathbf{1}_A(\mu)\Big)=1.
\]
\item For every $\mu\in(\R\setminus\Q)\setminus N_0$ (where $N_0$ is as in Theorem~\ref{thm:coverQ}),
\[
\Prob_\mu\!\Big(\lim_{n\to\infty} F_n(X_1,\dots,X_n)=0\Big)=1.
\]
\end{enumerate}
Equivalently, the test makes only finitely many mistakes almost surely for every rational mean and for every
irrational mean outside $N_0$.
\end{theorem}

\begin{proof}
Run Cover's procedure for $S=\Q$ on the sample stream (Theorem~\ref{thm:coverQ}).
Let $C_n\in\{0,1,2,\dots\}$ denote its stage-$n$ output, where $C_n=0$ means ``$\mu\notin\Q$'' and $C_n=i\ge 1$
means ``$\mu=q_i$''.

Define
\[
F_n(x_1,\dots,x_n) \;:=\;
\begin{cases}
0, & C_n(x_1,\dots,x_n)=0,\\[4pt]
a(C_n(x_1,\dots,x_n),\,n), & C_n(x_1,\dots,x_n)\ge 1.
\end{cases}
\]

\smallskip\noindent
\emph{Case 1: $\mu\in\Q$.}
Then $\mu=q_{i^\ast}$ for some index $i^\ast$.
By Theorem~\ref{thm:coverQ}(i), with probability one there exists $N$ such that for all $n\ge N$,
$C_n(X_1,\dots,X_n)=i^\ast$.
On that event,
\[
F_n(X_1,\dots,X_n) = a(i^\ast,n) \quad\text{for all }n\ge N.
\]
Since $a(i^\ast,n)$ stabilizes to $\mathbf{1}_{I_A}(i^\ast)=\mathbf{1}_A(\mu)$, $F_n$ stabilizes to the correct value.

\smallskip\noindent
\emph{Case 2: $\mu\notin\Q$ and $\mu\notin N_0$.}
By Theorem~\ref{thm:coverQ}(ii), with probability one there exists $N$ such that for all $n\ge N$,
$C_n(X_1,\dots,X_n)=0$.
Hence $F_n(X_1,\dots,X_n)=0$ for all $n\ge N$ and $F_n$ stabilizes to $0$.
Since $A\subseteq\Q$, this is correct.

Thus (i)-(ii) hold.
\end{proof}

\begin{corollary}[Computable subsets of $\Q$]\label{cor:computable}
If $A\subseteq\Q$ is computable (decidable) relative to the enumeration $e$, then the conclusion of
Theorem~\ref{thm:main} holds with $a(i,s)\equiv \mathbf{1}_{I_A}(i)$ and no appeal to Lemma~\ref{lem:shoenfield}.
\end{corollary}

\begin{theorem}[Necessity: limit tests force $\Delta^0_2$ index sets]\label{thm:necessity}
Fix a computable one-to-one enumeration $e:\N\to\Q$ and let $A\subseteq\Q$ with index set
$I_A=\{i\in\N: e(i)\in A\}$.
Suppose there exists a computable sequential test $(F_n)$ such that for every rational $q\in\Q$,
whenever $(X_n)_{n\ge 1}$ is i.i.d.\ with $\E[X_1]=q$,
\[
\Prob_q\!\Big(\lim_{n\to\infty} F_n(X_1,\dots,X_n)=\mathbf{1}_A(q)\Big)=1.
\]
Then $I_A$ is limit computable; in particular, $I_A\in\Delta^0_2$.
Consequently, if $I_A\notin\Delta^0_2$, no such computable limit test exists.
\end{theorem}
\begin{proof}
Fix $i\in\N$ and write $q=e(i)$.
Consider the degenerate i.i.d.\ process $X_n\equiv q$ (so $\E[X_1]=q$ and $\Prob_q$ is supported on the constant sample path).
By hypothesis, along this process the sequence $F_n(q,\dots,q)$ stabilizes almost surely to $\mathbf{1}_A(q)$.
Define
\[
a(i,n)\;:=\;F_n(\underbrace{q,\dots,q}_{n\ \text{times}}).
\]
Since $e$ and $(F_n)$ are computable, the map $(i,n)\mapsto a(i,n)$ is total computable, and
\(\lim_{n\to\infty} a(i,n)=\mathbf{1}_A(e(i))=\mathbf{1}_{I_A}(i)\).
Thus $I_A$ is limit computable.
By Lemma~\ref{lem:shoenfield}, this implies $I_A\in\Delta^0_2$.
The final claim is the contrapositive.
\end{proof}
\noindent
Combining Theorems~\ref{thm:main}, and \ref{thm:necessity} yields the announced characterization: for rational means, computable limit-testability of $A\subseteq\mathbb{Q}$ is equivalent to limit computability of $I_A$.
\begin{remark}[Exceptional sets]\label{rem:exceptional-sets}
Cover-style \emph{finite-error} sequential rules typically come with two qualitatively different guarantees:
\begin{itemize}
\item \emph{On-target correctness:} for every parameter in the target family one has eventual correctness
(almost surely, with only finitely many errors).
\item \emph{Off-target correctness:} outside the target family one can usually guarantee eventual correctness
only up to a Lebesgue-null exceptional set.
\end{itemize}
In particular, already in Cover's original construction (and hence already for $S=\Q$),
the induced irrationality test is eventually correct for all $\mu\in\Q$, and it is also eventually correct for
$\mu\in\R\setminus\Q$ \emph{except possibly} on a null set $N_0\subseteq \R\setminus\Q$; see~\cite{Cover1973}.
Equivalently, one achieves finite-error discrimination between $\Q$ and a full-measure subset of $\R\setminus\Q$,
but not necessarily all irrationals.

Our $\Delta^0_2$ lift does \emph{not} worsen this phenomenon.
The additional $\Delta^0_2$ layer is a purely \emph{effective} post-processing of the stabilized index produced by
Cover's step on $\Q$ (via the given $\Delta^0_2$/Shoenfield representation of $A$).
Thus, for every rational mean $\mu\in\Q$ the decision is eventually correct (hence, in particular, for all $\mu\in A$
and all $\mu\in\Q\setminus A$), and for irrational means any possible failure can occur only on the \emph{same}
null exceptional set $N_0$ already present in the underlying Cover identification step.
In this sense, the $\Delta^0_2$ lift is computably sharp: it preserves Cover's finite-error guarantees while extending
them uniformly to all $\Delta^0_2$ subsets of $\Q$ under an explicit effective representation.
\end{remark}

\begin{remark}[A concrete example]
Let $K\subseteq\N$ be the halting set. Then $K$ is $\Delta^0_2$ (actually, $\Sigma^0_1$) and the set
$A:=\{q_i : i\in K\}\subseteq\Q$ is $\Delta^0_2$ relative to $e$.
Theorem~\ref{thm:main} yields a sequential procedure which, when the mean happens to be rational,
eventually answers whether its \emph{index} lies in $K$.
In other words, probabilistic identification of the rational mean can be composed with highly nontrivial
limit computations on the associated index.
\end{remark}

\section{From $\Q$ to effectively enumerated sets of computable reals}

Cover's construction \cite{Cover1973} is formulated for an arbitrary countable set
$S=\{s_1,s_2,\dots\}\subseteq\R$.
To regard the resulting decision rule as \emph{computable} for general $S$, one must specify an effective
presentation of the reals $s_j$ and an effective analogue of the ``least index in a shrinking interval'' step.
For $\Q$ this is immediate from any computable enumeration.

\paragraph{Effective enumeration via computable Cauchy bounds.}
We assume that $S$ is given by total computable functions
$L,U:\N\times\N\to\Q$ such that for all $j,m\ge 1$,
\begin{equation}\label{eq:cauchy-bounds}
L(j,m)\le s_j\le U(j,m)\quad\text{and}\quad U(j,m)-L(j,m)\le 2^{-m}.
\end{equation}
(Equivalently, each $s_j$ is given by a computable Cauchy name, uniformly in $j$.)
Using open intervals, define for rationals $x$ and $\delta>0$ the \emph{certified inclusion} predicate
\[
\mathrm{In}(j;x,\delta)\;:\Longleftrightarrow\;\exists m\ \bigl(x-\delta<L(j,m)\ \wedge\ U(j,m)<x+\delta\bigr).
\]
If $\mathrm{In}(j;x,\delta)$ holds then $s_j\in(x-\delta,x+\delta)$.
For a stage parameter $n$ we also define the bounded stage-$n$ version
\[
\mathrm{In}_n(j;x,\delta)\;:\Longleftrightarrow\;\exists m\le n\ \bigl(x-\delta<L(j,m)\ \wedge\ U(j,m)<x+\delta\bigr),
\]
which is decidable uniformly in $(j,n,x,\delta)$.

\paragraph{A computable bounded ``least index'' subroutine.}
Given integers $k,n\ge 1$ and rationals $x,\delta>0$, define
\begin{equation}\label{eq:iSk}
i^{(n)}_{S,k}(x,\delta)\;=\;\min\{\,1\le j\le k:\ \mathrm{In}_n(j;x,\delta)\,\},
\end{equation}
with the convention that $i^{(n)}_{S,k}(x,\delta)=0$ if no such $j\le k$ exists.
Then $i^{(n)}_{S,k}$ is total computable.
Moreover, if $s_j\in(x-\delta,x+\delta)$ and the inclusion is strict, then for all sufficiently large $n$ we have
$\mathrm{In}_n(j;x,\delta)$.

\medskip
The following theorem records that once one has a computable identification rule for $S$ (in the sense of Cover),
any limit-computable labeling of the indices yields a computable Popper-style test by the same Shoenfield
composition as in the rational case.
\begin{lemma}\label{lem:iSk-computable}
For each $k,n\ge 1$ and rationals $x,\delta>0$, the predicate $\mathrm{In}_n(j;x,\delta)$ is decidable uniformly
in $(j,n,x,\delta)$, and the map $i^{(n)}_{S,k}(x,\delta)\in\{0,1,\dots,k\}$ defined in \eqref{eq:iSk} is total
computable.
Moreover, if $s_j\in(x-\delta,x+\delta)$ with strict inclusion, then $\mathrm{In}_n(j;x,\delta)$ holds for all
sufficiently large $n$.
\end{lemma}
A complementary eventual-exclusion statement is recorded in Appendix~\ref{app:identifier}, Lemma~\ref{lem:eventual-certification}.

\begin{proof}
Decidability of $\mathrm{In}_n$ is immediate since it is a finite existential quantifier over $m\le n$ and all
quantities are rational.
Total computability of $i^{(n)}_{S,k}$ follows by finite search over $1\le j\le k$ and returning the least $j$
satisfying $\mathrm{In}_n(j;x,\delta)$, or $0$ if none does.
For the final claim, strict inclusion gives $\varepsilon>0$ with
$x-\delta+\varepsilon \le s_j \le x+\delta-\varepsilon$.
Choosing $m$ with $2^{-m}<\varepsilon$ and using \eqref{eq:cauchy-bounds}, one has
$x-\delta < L(j,m)$ and $U(j,m)<x+\delta$, hence $\mathrm{In}_n(j;x,\delta)$ holds for all $n\ge m$.
\end{proof}

Under the above effective presentation assumptions on $S=\{s_1,s_2,\dots\}\subseteq\mathbb{R}$,
the subroutine
\[
i^{(n)}_{S,k}(x,\delta)\;=\;\min\{\,1\le j\le k:\ \mathrm{In}_n(j;x,\delta)\,\}
\]
is computable (uniformly in $n,k,x,\delta$), and hence Cover's sequential identifier on $S$
admits a computable implementation.
A detailed proof is given in the Appendix.
Consequently, in what follows, we may treat Cover's identifier on such sets $S$ as a computable procedure, and therefore the same $\Delta^0_2$ characterization of finite-error sequential
membership tests extends verbatim to all subsets $A\subseteq S$.

\begin{theorem}[Countable $S$ and $\Delta^0_2$ labeling]\label{thm:generalS}
Let $S=\{s_1,s_2,\dots\}\subseteq\R$ be countable and presented by computable Cauchy bounds
$L,U:\N\times\N\to\Q$ satisfying \eqref{eq:cauchy-bounds}.
Fix a set $A\subseteq S$ and its index set
\[
I_A:=\{\,j\in\N:\ s_j\in A\,\}\subseteq\N.
\]
Assume that $I_A\in\Delta^0_2$, witnessed by a total computable approximation
$a:\N\times\N\to\{0,1\}$ such that $\mathbf{1}_{I_A}(j)=\lim_{n\to\infty} a(j,n)$ for all $j$.

Assume further that there exists a computable sequential rule
\[
C:\Q^{<\omega}\to\{0,1,2,\dots\}
\]
such that for i.i.d.\ observations with finite second moment and mean $\mu$:
\begin{enumerate}
\item If $\mu\in S$, then with probability one there exist $j\ge 1$ with $s_j=\mu$ and $N$ such that
$C(X_1,\dots,X_n)=j$ for all $n\ge N$.
\item If $\mu\notin S$, then for all $\mu\in(\R\setminus S)\setminus N_0$ (where $N_0\subseteq \R\setminus S$ is
Lebesgue-null), with probability one there exists $N$ such that $C(X_1,\dots,X_n)=0$ for all $n\ge N$.
\end{enumerate}
Then there exists a computable sequential test $F:\Q^{<\omega}\to\{0,1\}$ such that:
\begin{enumerate}
\item If $\mu\in S$, then $F(X_1,\dots,X_n)\to \mathbf{1}_A(\mu)$ almost surely.
\item If $\mu\in(\R\setminus S)\setminus N_0$, then $F(X_1,\dots,X_n)\to 0$ almost surely.
\end{enumerate}
\end{theorem}

\begin{proof}
Since $I_A$ is limit computable, there exists a total computable function $a:\N\times\N\to\{0,1\}$
such that for each $j\ge 1$, $\lim_{n\to\infty} a(j,n)=\mathbf{1}_{I_A}(j)$.
Define $F:\Q^{<\omega}\to\{0,1\}$ by
\[
F(X_1,\dots,X_n)=
\begin{cases}
a(C(X_1,\dots,X_n),\,n), & \text{if } C(X_1,\dots,X_n)\ge 1,\\
0, & \text{if } C(X_1,\dots,X_n)=0.
\end{cases}
\]
This $F$ is computable as a composition of computable maps.

If $\mu\in S$, then almost surely there exist $j\ge 1$ and $N$ such that $C(X_1,\dots,X_n)=j$ for all $n\ge N$.
Hence $F(X_1,\dots,X_n)=a(j,n)$ for all $n\ge N$, so
\[
\lim_{n\to\infty}F(X_1,\dots,X_n)=\lim_{n\to\infty}a(j,n)=\mathbf{1}_{I_A}(j)=\mathbf{1}_A(\mu)
\quad\text{almost surely.}
\]
If $\mu\in(\R\setminus S)\setminus N_0$, then almost surely $C(X_1,\dots,X_n)=0$ for all large $n$,
hence $F(X_1,\dots,X_n)=0$ eventually and $\lim_{n\to\infty}F(X_1,\dots,X_n)=0$ almost surely.
\end{proof}
\begin{remark}[Necessity on $S$]
If a computable sequential rule $F$ stabilizes correctly for every mean $\mu\in A\subseteq S$, then by evaluating $F$ on the
degenerate i.i.d.\ process $X_n\equiv s_j$ one obtains a computable approximation
$a(j,n)=F(s_j,\dots,s_j)$ converging to $\mathbf{1}_{I_A}(j)$.
By Shoenfield's limit lemma this forces $I_A$ to be  $\Delta^0_2$.
\end{remark}
Finally, we also note that if we restrict to effectively presented sets of reals, there is no need for the null set of failures, since the computable reals are countable. The proof is identical to the above discussion, ignoring the part $\mu\notin S$.
\section{Conclusion: finite-error inquiry and Popperian themes}\label{sec:conclusion}

Cover's theorem is striking because it separates two kinds of success one might demand of a
data-driven method.  If we insist on \emph{uniform} guarantees at a fixed sample size, even simple
questions about a mean quickly become impossible.  By contrast, if we permit \emph{revisions} and
demand only \emph{eventual correctness} that only finitely many wrong decisions occur almost surely
along an infinite data stream then remarkably strong forms of learnability reappear. This type of inference exists in our studies of physics, where theories might be refuted, given new evidence, but hopefully eventually become correct.

The main contribution of this paper is to locate the precise boundary of this phenomenon under computability constraints. Indeed, as far as we know, the physical Church-Turing thesis holds, and any sequence of physical experiments yields a sequence of computable outcomes.

For rational means, we give a complete characterization of which rational subsets (
$A\subseteq \Q$) admit a computable sequential membership procedure with only finitely many errors
almost surely.  More generally, for any countable hypothesis class of means equipped with an
effective presentation sufficient to implement Cover's identification rule, we obtain the same kind
of necessary and sufficient criterion.  In this sense, the results provide an \emph{optimal
generalization} of Cover's theorem: within the natural effectivity assumptions required to make the
procedure computable, we fully describe which hypotheses are amenable to finite-error sequential
testing.

The notion of finite-error success formalizes a style of reasoning common in both statistics and scientific practice: provisional commitment, followed by the willingness to retract and revise in
light of new evidence.  A procedure that is eventually correct is not required to be right quickly,
and it need not provide a certified stopping time at which it can announce finality.  Instead, it
models inquiry as a process whose \emph{trajectory} stabilizes: after some (random) point, the
procedure's verdict does not change again and coincides with the truth almost surely.

This perspective clarifies why countability and effectivity enter naturally.  When hypotheses are
countable and presented effectively, we can search among them in a manner compatible with
computation, and the data can asymptotically ``select'' the correct hypothesis.  Without an effective
presentation, even if the hypothesis class is countable, there need not exist a
computable mechanism that can exploit that countability.

\subsection{Popper, falsification, and convergence to the truth}
Popper emphasized that scientific theories are not verified but are subjected to severe tests and
may be \emph{falsified}.  A sequential procedure that is allowed to change its mind fits naturally
into this picture: early conjectures can be refuted by accumulating evidence, leading to revised
conjectures, and so on.  What our results add is a sharp sense in which \emph{falsification-driven
revision can be made effective} in a probabilistic setting: for exactly those hypothesis sets
identified by our characterization, there exists a computable method that will eventually cease to
be falsified, in the sense that it will make only finitely many incorrect commitments almost surely.

At the same time, the theorems delineate a principled limitation on Popperian optimism.  Even when
one relaxes demands to permit infinitely many revisions in principle (while requiring only finitely
many errors almost surely), not every property of the mean is learnable by a computable sequential
method.  Thus, the slogan ``science converges to the truth'' becomes a mathematically constrained
claim: convergence is attainable, but only for those targets lying on the right side of the boundary
identified here, and only under the effectivity conditions that make the hypothesis class accessible
to computation.
One interesting extension is to characterize other statistical estimation and decision problems under the same guiding question: which hypotheses admit \emph{computable} inquiry that is allowed to err only finitely often?
The philosophical moral remains the same: permitting revision expands what can be learned, but computability
and representation impose sharp, informative limits.
\section{Use of AI}
OpenAI's ChatGPT5.4 thinking mode was used as a writing aid to suggest alternative phrasing, improve exposition,
and assist with LaTeX editing and consistency checks. All definitions, theorems, proofs, and bibliographic choices were developed and verified by the author(s), who take full responsibility for the correctness and originality of the results. The language model also wrote a critical review of the final version. This was used to improve the quality of the final version.

\bibliographystyle{ieeetr}
\bibliography{references}
\appendix
\section{Uniform presentations and computable identifiers for countable mean classes}\label{app:identifier}

This appendix isolates the effectivity conditions required to implement Cover's countable-mean identification rule when the hypothesis class $S=\{s_1,s_2,\dots\}\subseteq \R$ is given effectively. It uses the definitions of effective presentation,

\subsection{Uniform Cauchy bounds and certified inclusion}

We work with the presentation already used in \eqref{eq:cauchy-bounds}: total computable $L,U:\N\times\N\to\Q$
satisfying for all $j,m\ge 1$,
\[
L(j,m)\le s_j\le U(j,m)
\quad\text{and}\quad
U(j,m)-L(j,m)\le 2^{-m}.
\]
(Optionally, one may enforce nestedness by intersecting successive intervals.)

For rationals $x$ and $\delta>0$, define the stage-$n$ certified inclusion predicate
\begin{equation}
\label{def:In}
\mathrm{In}_n(j;x,\delta)\ :\Longleftrightarrow\ \exists m\le n\ \bigl(x-\delta<L(j,m)\ \wedge\ U(j,m)<x+\delta\bigr),
\end{equation}
as in Section~7.

Given $k,n\ge 1$, define the bounded least-index operator
\[
i^{(n)}_{S,k}(x,\delta)=\min\{1\le j\le k:\ \mathrm{In}_n(j;x,\delta)\},
\]
with the convention that it is $0$ if the set is empty.

\begin{lemma}[Effectivity of the least-index step]\label{lem:least-index-effective}
For each fixed $k,n\ge 1$, the predicate $\mathrm{In}_n(j;x,\delta)$ is decidable uniformly in
$(j,x,\delta)$ on rational inputs, and the map $i^{(n)}_{S,k}(x,\delta)\in\{0,1,\dots,k\}$ is total
computable uniformly in $(k,n,x,\delta)$.
\end{lemma}

\begin{proof}
Decidability of $\mathrm{In}_n$ is immediate: it is a finite existential quantifier over $m\le n$ and
all comparisons are between rationals. Then $i^{(n)}_{S,k}$ is computed by finite search over
$j=1,\dots,k$ and returning the least $j$ satisfying $\mathrm{In}_n$, or $0$ if none does.
\end{proof}

\begin{lemma}[Eventual correctness under strict inclusion]\label{lem:eventual-certification}
Fix $j\ge 1$, $x\in\R$, and $\delta>0$.
\begin{enumerate}
\item If $s_j\in(x-\delta,x+\delta)$, then $\mathrm{In}_n(j;x,\delta)$ holds for all sufficiently large $n$.
\item If $s_j\notin[x-\delta,x+\delta]$, then $\mathrm{In}_n(j;x,\delta)$ fails for all sufficiently large $n$.
\end{enumerate}
\end{lemma}

\begin{proof}
In \eqref{eq:delta}, strict inclusion gives $\varepsilon>0$ with $x-\delta+\varepsilon\le s_j\le x+\delta-\varepsilon$.
Choose $m$ so that $2^{-m}<\varepsilon$. Then $L(j,m)>x-\delta$ and $U(j,m)<x+\delta$, hence
$\mathrm{In}_n(j;x,\delta)$ holds for all $n\ge m$.
In \eqref{eq:lil-inclusion}, the distance from $s_j$ to the closed interval $[x-\delta,x+\delta]$ is positive, so for all
large $m$ the approximation interval $[L(j,m),U(j,m)]$ lies outside $(x-\delta,x+\delta)$ and cannot
be contained in it; hence $\mathrm{In}_n$ eventually fails.
\end{proof}

\subsection{From Cover's rule to a computable identifier}

In the body of the paper, Cover's specialization to $\Q$ uses the least-index function
$i(t,\delta)=\min\{i:|q_i-t|<\delta\}$. For a general $S$ with presentation $(L,U)$, the effective
replacement is the bounded operator $i^{(n)}_{S,k}(x,\delta)$ above.

To incorporate finite-precision readouts from Section~5, we use the inflated tolerance
$\delta_n'=\delta_n+\eta_n$ (cf.~\eqref{eq:delta-prime}), so that mean comparisons remain correct under the deterministic
error bound \eqref{eq:mean-bound-fixed}).

\begin{definition}[Computable identifier for $S$]\label{def:computable-identifier}
A computable identifier for $S$ is a sequence of computable maps
$C_n:\Q^n\to\{0,1,2,\dots\}$ such that there exists a Lebesgue-null set $N_0\subseteq\R\setminus S$
with:
\begin{enumerate}
\item If $\mu=s_j\in S$, then $C_n(\tilde X_1,\dots,\tilde X_n)=j$ for all large $n$, almost surely.
\item If $\mu\in(\R\setminus S)\setminus N_0$, then $C_n(\tilde X_1,\dots,\tilde X_n)=0$ for all large $n$,
almost surely.
\end{enumerate}
\end{definition}

\begin{theorem}[Effective implementation of Cover's identifier]\label{thm:identifier-exists}
Let $S=\{s_1,s_2,\dots\}\subseteq\R$ be countable and presented by computable Cauchy bounds $(L,U)$.
Fix schedules $(n(j))_{j\ge 1}$ and $(k_j)_{j\ge 1}$ as in Cover's construction, and define $\delta_n$
by \eqref{eq:delta} and $\delta_n'=\delta_n+\eta_n$ by \eqref{eq:delta}.
Assume the corresponding Borel--Cantelli summability condition holds with $\delta'$ (equivalently,
$\sum_j k_j\,\delta'_{n(j)}<\infty$, up to constants as in \eqref{eq:borel-cantelli-sum}.
Define at decision times $n(j)$:
\[
C_{n(j)}(\tilde X_1,\dots,\tilde X_{n(j)})
:= i^{(m(j))}_{S,k_j}\!\bigl(\overline{\tilde X}_{n(j)},\,\delta'_{n(j)}\bigr),
\]
for any computable schedule $m(j)\to\infty$ (bounding the search depth in $\mathrm{In}_{m(j)}$), and keep $C_n$ constant between decision times.
Then $(C_n)$ is a computable identifier for $S$ (Definition~\ref{def:computable-identifier}).
\end{theorem}

\begin{proof}[Proof sketch]
Computability follows from Lemma~\ref{lem:least-index-effective} since all inputs are rational.
For correctness, Cover's probabilistic argument (LIL control plus Borel--Cantelli using \eqref{eq:borel-cantelli-sum}
shows that outside a Lebesgue-null exceptional set $N_0\subseteq\R\setminus S$,
the (ideal) least-index choice based on the true mean and true inclusion in shrinking open intervals
stabilizes as required. Replacing $\bar X_n$ by the readout mean $\overline{\tilde X}_n$ is harmless
once $|\overline{\tilde X}_n-\bar X_n|\le\eta_n$, by the inflation $\delta'\!=\delta+\eta$.
Finally, replacing true inclusion by certified inclusion does not change the eventual least index
whenever the target point lies strictly inside the relevant interval; this is exactly
Lemma~\ref{lem:eventual-certification}. Thus the effective selector agrees with the ideal selector
for all sufficiently large decision times on a probability-one set, yielding stabilization to the
correct index in $S$ and to $0$ off $S$ outside $N_0$.
\end{proof}

\begin{remark}
Theorem~4 assumes the existence of a computable identifier $C$ for $S$.
Theorem~\ref{thm:identifier-exists} provides such a $C$ whenever $S$ is uniformly presented by
computable Cauchy bounds \eqref{eq:cauchy-bounds}, i.e.\ exactly the setting in which Cover's ``least index'' step can be
implemented effectively.
\end{remark}

\end{document}